\documentclass[10pt]{article}
\usepackage{array,amsmath,amsfonts,graphicx,hyperref,mathrsfs,amsthm}

\newcommand{\bs}{\boldsymbol}

\newtheorem{theorem}{Theorem}

\newtheorem{proposition}{Proposition}

\begin{document}

\centerline{\Large\bf Bayesian model selection and averaging via mixture model estimation}

\vspace{2pt}

\vspace{.4cm}
\centerline{Merlin Keller$^1$ and Kaniav Kamary$^{1,2}$}
\vspace{.4cm}
\centerline{\it$^1$ EDF R\&D, Chatou, \it$^2$ AgroParisTech, Paris}

\vspace{.55cm}
\fontsize{9}{11.5pt plus.8pt minus .6pt}\selectfont

\begin{abstract}

A new approach for Bayesian model averaging (BMA) and selection is proposed, based on the mixture model approach for hypothesis testing in \cite{Kamari2014}. Inheriting from the good properties of this approach, it extends BMA to cases where improper priors are chosen for parameters that are common to all candidate models.

From an algorithmic point of view, our approach consists in sampling from the posterior distribution of the single-datum mixture of all candidate models, weighted by their prior probabilities. We show that this posterior distribution is equal to the `Bayesian-model averaged' posterior distribution over all candidate models, weighted by their posterior probability. From this BMA posterior sample, a simple Monte-Carlo estimate of each model's posterior probability is derived, as well as importance sampling estimates for expectations under each model's posterior distribution. 


\end{abstract}

\section{Introduction}

From a Bayesian point of view, hypothesis testing and, more generally, model selection and averaging \cite{Hoeting99,Wasserman00b}, usually involve the calculation of Bayes factors \cite{Kass95}, a notoriously difficult task, which has generated much literature over the last decades. Most of the methods which have been proposed to date rely on evaluating the marginal likelihood of each considered model, a quantity well defined only in presence of a proper prior, whose choice can be an issue when only weak prior information is available. 

Many marginal likelihood calculation methods have been proposed, including asymptotic approximations such as: the Bayesian information criterion (BIC) \cite{Schwarz1978} or  Laplace's approximation, importance sampling (all of which are reviewed in \cite{Kass95}); particle filtering \cite{Chopin2002}, exploiting the output of posterior sampling algorithms, through the harmonic mean identity \cite{Raftery2007}, the basic marginal equality \cite{Chib95,Chib2001}, or the divergence information criterion \cite{Spiegelhalter2002}, to cite just a few. Meanwhile, alternatives to the Bayes factor have been proposed to accomodate for improper priors, such as the {\em fractional} \cite{OHagan95} or {\em intrinsinc} Bayes factors \cite{Berger96}, which both divide the dataset into two parts; the first is used to update the improper prior into a proper posterior, which is then used as a prior for model selection, using the second part of the data.

Other methods avoid marginal likelihood evaluation, such as bridge sampling \cite{Meng1996}, which estimates the ratio of marginal likelihoods directly, or reversible-jump Monte-Carlo Markov Chain (MCMC) \cite{Green95}, which samples jointly the posterior distribution of the model indicator and model specific parameters, allowing for so-called trans-dimensional `jumps' between models. Alternatively, \cite{Carlin95} proposed to use a more classical MCMC approach to sample from the joint posterior distribution of the full vector parameter (obtained by concatenating the parameters from all candidate models), augmented by a discrete variable indicating the data generating model. See also \cite{Han01} for a review of other MCMC-based methods for Bayesian model choice. More recently, a solution based on Approximate Bayesian Computation (ABC) has been proposed in \cite{Grelaud09}, for models in which the likelihood function is not available. Similarly to above approaches, this consists in deriving the joint posterior distribution of model-specific parameters and model indicator, from which posterior model probabilities are then derived. However, this approach has been found unreliable in many cases \cite{Robert2011}, due to the intermodel insufficiency of the summary statistics used in the medtho. A solution has been found in \cite{Pudlo16}, recasting model choice as a classification problem, addressed by random forest techniques.

Going back to likelihood-based inference, a new and refreshing perspective on model choice can be found in the seminal work of \cite{Kamari2014}, which reveals its deep connections to mixture modeling in the following way: the mixture of all candidate models is considered, and its posterior distribution given the data derived. Model selection is then driven by the posterior mean, or median (or any other central quantity) of mixture weights. Recent theoretical results in \cite{Rousseau2011}, show that consistent model selection can be achieved in this fashion.

This solution has several key advantages over classical Bayesian model selection. In particular, it naturally allows the use of improper priors for parameters common to all considered models, provided the posterior mixture distribution is proper. From an algorithmic point of view, the problematic Bayes factor computation is replaced by mixture model estimation, which is arguably easier to perform, and in a more generic way, using MCMC samplers \cite{Marin2005}. See also \cite{Robert2015} for a detailed discussion on the difficulties associated with Bayes factors, and how these can be solved by mixture modeling.

In this paper, we use a simple argument to show that classical Bayesian model averaging can be formulated as a special case of mixture modeling, by considering the mixture of all candidate models, weighted by their prior probabilities. Using this reformulation, we are then able to benefit from all the advantages of the mixture modeling viewpoint, while remaining in the classical Bayesian model selection setting. This leads to a simple generic MCMC algorithm to perform Bayesian model averaging, from which posterior model probabilities and posterior samples within each candidate models are easily derived. Furthermore, it allows to share parameters accross models, which in turns naturally enables the use of improper priors.


The rest of this paper is organized as follows. In Section~\ref{sec:mixture}, we review classical Bayesian model selection and averaging, together with the mixture modeling approach to hypothesis testing in \cite{Kamari2014}. In Section~\ref{sec:BMA}, we show that Bayesian model averaging, hence also Bayesian model selection, is equivalent to mixture modeling in the single-datum setting. Based on this result, we introduce a new algorithm for Bayesian model selection and averaging, which consists in sampling the BMA posterior using a single MCMC procedure. From this, simple Monte-Carlo estimates of the posterior probabilities for each considered model can then be deduced, together with importance sampling reconstructions of the model-specific posteriors. The workings of our algorithm is illustrated in Section~\ref{sec:Illustration}.
 We conclude by a brief discussion in Section~\ref{sec:discussion}.

\section{Bayesian Model selection and averaging vs. Mixture modeling}\label{sec:mixture}

Consider the task of choosing a statistical model $\mathfrak M$ for a dataset $y$ among a finite list of candidate parametric models $\{\mathfrak M_k\}_{1\leq i\leq N}$, described by their likelihood function:
\begin{eqnarray*}
\mathfrak M_k &=& \{f_k(y | \theta), \theta \in \Theta\},
\end{eqnarray*}
Here we denote, without loss of generality, $\theta \in \mathbb R^p$ the set of all parameters. For instance, $\theta$ can be a single parameter, common to all models, or the concatenation $\theta = (\theta_1, \ldots, \theta_N)$ of model-specific parameters, with $f_k(y | \theta)=f_k(y | \theta_k)$ for $k=1,\ldots,N$. But any other configuration is conceivable. Adopting a Bayesian point of view, a prior density $\pi(\theta)$ is defined for $\theta$.

\subsection{Bayesian model selection} 

Classical Bayesian model selection then consists in considering the unknown model $\mathfrak M$ itself as a parameter to be estimated. Hence a prior distribution is defined on the set $\{\mathfrak M_k\}_{1\leq i\leq N}$ of candidate models:
\begin{eqnarray*}
\pi(\mathfrak M_k) &=& p_k,
\end{eqnarray*}
where $p_k$ is the prior probability that the data generating model $\mathfrak M$ is equal to $\mathfrak M_k$.

Applying Bayes' theorem, the posterior probability of model $\mathfrak M_k$ is then given by:
\begin{eqnarray}\label{eq:posterior_probabilities}
\pi(\mathfrak M_k | y) &=& \frac{p_k m_k(y)}{\sum_j p_j m_j(y)},
\end{eqnarray}
where 
\begin{eqnarray}\label{eq:marginal_likelihood}
m_k(y) &=& \int f_k(y|\theta) \pi(\theta) d\theta
\end{eqnarray}
is the {\em marginal likelihood}, or {\em evidence}, for model $\mathfrak M_k.$ From (\ref{eq:posterior_probabilities}), note that the actual value of $m_k(y)$ for each model $k$ need only be known up to a common multiplicative constant, since the posterior probabilities are defined in terms of their ratios.

However, most approaches to Bayesian model selection focus on computing $m_k(y)$ for each model $\mathfrak M_k$, in order to derive the posterior probabilities (\ref{eq:posterior_probabilities}). This forces $\pi(\theta)$ to be a proper distribution, otherwise $m_k(y)$ is not well-defined. This  is a major issue, since the influence of the choice of a prior on the outcome of the selection process is important (much more so than for estimation), and not yet well understood. 

After having selected the {\em a posteriori} most probable model, that is, the model $\mathfrak M_k$ which maximizes $p_k m_k(y)$, Bayesian inference is then performed based on the selected model's posterior distribution:
\begin{eqnarray}\label{eq:posterior_density}
\pi(\theta|y,\mathfrak M_k) &=& \frac{f_k(y|\theta)\pi(\theta)}{m_k(y)}.
\end{eqnarray}
This is typically sampled using Monte-Carlo Markov chain (MCMC) techniques \cite{Robert04}.

\subsection{Bayesian model averaging.} 

The above procedure works reasonably well if one model is significantly more probable than all the others, otherwise it fails to properly account for model uncertainty. Hence, an alternative consists in integrating out $\mathfrak M$ from the joint posterior distribution of $(\theta, \mathfrak M)$. From (\ref{eq:posterior_probabilities}) and (\ref{eq:posterior_density}), the resulting `model-averaged' posterior for $\theta$ is seen to be equal to:
\begin{eqnarray}\label{eq:BMA}
\pi(\theta | y) &=& \int  \pi(\theta | y, \mathfrak M) \pi(\mathfrak M | y) d\mathfrak M \nonumber\\
&\propto& \sum_k p_k m_k(y) \pi(\theta | y, \mathfrak M_k),
\end{eqnarray}
that is, the average of the posterior densities for all models, weighted by their posterior probabilities. This forms the basis of the Bayesian model averaging (BMA) approach \cite{Hoeting99}, which usually consists of the following steps:
\begin{enumerate}
\item\label{step:marginal} Compute the marginal likelihood $m_k(y)$ for all models, and deduce posterior probabilities (\ref{eq:posterior_probabilities});
\item\label{step:filter} Remove all models which have negligible posterior probabilities;
\item\label{step:posterior} Compute the posterior density for all remaining models, in order to form the model-averaged posterior (\ref{eq:BMA})
\end{enumerate}
Hence, BMA adds to the difficulty of computing the marginal likelihoods, inherent to Bayesian model selection, that of having to perform Bayesian inference separately within each candidate model which hasn't been ruled out in the second step. These difficulties explain why this approach is not yet widely spread within praticians.

\subsection{Model selection through mixture model estimation} 

In \cite{Kamari2014}, an alternative to the classical framework reviewed above is proposed. It consists in embeding the competing models $\mathfrak M_k$ in an encompassing mixture model $\mathfrak M_{\bs \alpha}$ (slightly abusing notations), defined by:
\begin{eqnarray}\label{eq:mixture}
y | \bs\alpha, \theta, \mathfrak M_{\bs \alpha} &\sim& \sum_k \alpha_k f_k(y|\theta),
\end{eqnarray}
with $\alpha_k \in [0,1], \sum_k \alpha_k = 1.$ Hence, each model $\mathfrak M_k$ corresponds to the degenerate case where $\alpha_k = 1$, and all other mixture weights are null. Model selection is thus re-defined as the task of estimating the mixture weights $\bs\alpha$, from a sample ${\bf y} = (y_1, \ldots, y_n)$, having endowed the weights with a Dirichlet prior:
\begin{eqnarray}\label{eq:Dirichlet}
\pi(\bs \alpha|\bs a) &\propto& \prod_k \alpha_k^{(a_k - 1)}\bs 1_{\{0\leq \alpha_k\leq 1\}} \bs 1_{\{\sum_k \alpha_k =1\}}.
\end{eqnarray}
Indeed, owing to \cite{Rousseau2011}, it can be shown that if all the $y_i$ are distributed according to a single model $\mathfrak M_k$, then the posterior distribution of the corresponding mixture weight $\alpha_k$ concentrates around $1$ as $n$ goes to infinity. Hence, this approach shares the consistency of the classical Bayesian model selection framework (see for instance \cite{Schwarz1978}).

However, contrary to the regular approach, it allows to use an improper prior for $\theta$, as long as the posterior defined under the mixture model~(\ref{eq:mixture}) is proper. This is a huge advantage, considering the dependance of the marginal likelihood (\ref{eq:marginal_likelihood}) with respect to prior choice, and the impossiblity to use improper reference priors, except by using part of the data to update it first into a proper distribution, as in \cite{Berger96}.

Moreover, this approach avoids the problematic Bayes factor computation, replacing it by mixture model estimation, which can be performed using standard MCMC procedures \cite{Marin2005}. \cite{Kamari2014} advocates the use of a Metropolis-Hastings algorithm to sample the posterior distribution of $(\theta,\alpha)$, based on the product likelihood:
\begin{eqnarray}\label{eq:mixture_product}
{\bf y} | \bs\alpha, \theta &\sim& \prod_{i=1}^n \left\{ \sum_k \alpha_k f_k(y_i|\theta) \right\}.
\end{eqnarray}
In contrast, mixture modeling is commonly based on a data augmentation scheme, which adds latent class labels $\zeta_i \in \{1, \ldots, N\}$, such that:
\begin{eqnarray}\label{eq:mixture_expansion}
{\bf y} | {\bs \zeta}, \bs\alpha, \theta &\sim& \prod_{i=1}^n f_{\zeta_i}(y_i|\theta) \\
\bs \zeta | \bs\alpha &\sim& \mathcal M(n, \alpha_1, \ldots, \alpha_N ),
\end{eqnarray}
where $\mathcal M(n, \alpha_1, \ldots, \alpha_N )$ denotes the multinomial distribution with $N$ modalities and $n$ observations. However, 
the resulting Gibbs sampler is typically prone to high intra-chain autocorrelation, due to the dificulty of the labels $\zeta_i$ to `jump' from one discrete value to the other.

Beyond these technical considerations, numerical results confirm the ability of this approach to asymptotically recover the data-generating model, while remaining computationally attractive. In view of these desirable features, the mixture modeling approach introduced in \cite{Kamari2014} for model selection appears as a promising alternative to the usual Bayesian model selection paradigm. As we will now see, it turns out that Bayesian model selection is in fact already part of this mixture model framework.

\section{Bayesian model averaging as a single-datum mixture model estimation}\label{sec:BMA}


We now state our main result:

\begin{theorem}[BMA - mixture modeling inclusion]
The BMA posterior distribution (\ref{eq:BMA}) is equal to the posterior distribution of the mixture model (\ref{eq:mixture}), written in the special case of a single observation:
\begin{eqnarray}\label{eq:mixture_model_posterior}
\pi(\theta | y) &\propto&  \sum_k p_k f_k(y|\theta)\pi(\theta),
\end{eqnarray}
provided that the prior expectations of the mixture weights under the Dirichlet prior (\ref{eq:Dirichlet}), given by: $\mathbb E[\alpha_k]=\frac{a_k}{\sum_\ell a_\ell}$, are equal to $p_k.$

\end{theorem}

\begin{proof} 
Simply observe that the BMA posterior (\ref{eq:BMA}) can be re-written as:
\begin{eqnarray*}
\pi(\theta | y) &\propto& \sum_k p_k m_k(y) \pi(\theta | y, \mathfrak M_k)\nonumber\\
&\propto& \sum_k p_k m_k(y)  \frac{f_k(y|\theta)\pi(\theta)}{m_k(y)}\nonumber\\
&\propto& \sum_k p_k f_k(y|\theta) \pi(\theta).
\end{eqnarray*}
\end{proof}


This result suggests a new definition for the BMA posterior, according to (\ref{eq:mixture_model_posterior}). This has several advantages over the traditional formulation (\ref{eq:BMA}):

\begin{itemize} 

\item it does not depend on the marginal likelihoods $m_k(y).$ In particular, it remains well-defined even when the prior on shared parameters are improper, as long as the right member in (\ref{eq:mixture_model_posterior}) is integrable;

\item The model-averaged posterior can be sampled directly, as described in the next section. In this way, BMA can be implemented using a single MCMC algorithm, without computing marginal likelihoods or having to pre-select models;

\item The parameters of the different models are considered all at once, as part of a joint parameter space, rather than on separate spaces, according to the usual setting of Bayesian model selection/averaging. This is what enables to share parameters between models, reducing the computational cost of the approach and widening possible prior choices. It also avoids having to `jump' accross models of varying dimensions \cite{Green95}.

\end{itemize}

\subsection{A generic MCMC algorithm for Bayesian model selection and averaging}\label{sec:algo}


The simplest way to sample from the BMA posterior~(\ref{eq:mixture_model_posterior}) is to use the Metropolis-Hastings (MH) algorithm \cite{Robert04}:

\begin{enumerate}
\item Given a current value $\theta,$ generate a proposal $\theta^\ast$ from an instrumental distribution $q(\theta^\ast|\theta)$;
\item Compute the acceptance rate:
$$
\alpha = \min\left\{ \frac{ \sum_k p_k f_k(y|\theta^\ast)\times \pi(\theta^\ast) \times q(\theta|\theta^\ast) }{ \sum_k p_k f_k(y|\theta) \times \pi(\theta)\times  q(\theta^\ast|\theta) }, 1 \right\}
$$
\item Simulate $u\sim \mathcal U([0,1])$; if $u < \alpha$ then update current value: $\theta^\ast \to \theta$. Otherwise, keep $\theta$ as it is.
\end{enumerate}

A default choice for the instrumental distribution $q(\theta^\ast|\theta)$ when $\theta$ is defined on an open subset of $\mathbb R^p$ is the Gaussian random-walk: $\theta^\ast \sim \mathcal N(\theta; \Sigma)$, where the instrumental covariance matrix $\Sigma$ is chosen to attain a `reasonable' acceptance rate, not too low in order to avoid being stuck in a single point, and not too high in order to explore the whole posterior distribution. Blocked MH, where subsets of the parameter vectors are updated in turn, is usually a good idea to promote chain mixing.

Note that this is similar, but different, from the MCMC algorithm introduced in \cite{Carlin95}, in which the joint distribution of $(\theta,\mathfrak M)$ (where $\mathfrak M$ is the unknown data-generating model) is sampled using a blocked Gibbs sampler. Here, we use the fact that $\mathfrak M$ can be integrated out analytically given $\theta$, which reduces the dimension and promotes better mixing of the Markov chain. Another difference is that in \cite{Carlin95}, parameters are not shared accross models, though their approach could easily be extended in that direction.

\paragraph{Posterior model probability.}

Once the BMA posterior $\pi(\theta|y)$ has been sampled, the posterior probability for model $\mathfrak M_k$ can easily be recovered as the posterior expectation:
\begin{eqnarray*}
\pi(\mathfrak M_k | y) &=& \int \pi(\mathfrak M_k | \theta, y) \pi(\theta | y)d\theta\\
&=& \int \frac{p_k f_k(y | \theta)}{ \sum_j p_j f_j(y | \theta)} \pi(\theta | y)d\theta.
\end{eqnarray*}


In practice, having generated a BMA posterior sample $(\theta_1, \ldots, \theta_S)$, the posterior probabilities are estimated as:
\begin{eqnarray}\label{eq:posterior_probability}
\widehat \pi(\mathfrak M_k|y) &=& 
\frac{1}{S} \sum_{s=1}^S 
\frac
{p_k f_k(y | \theta_s)}
{ \sum_j p_j f_j(y | \theta_s)}.
\end{eqnarray}
The Bayes factor $BF_{k\ell}(y) = \frac{m_k(y)}{m_\ell(y)}$, that is, the prior to posterior odds ratio for model $\mathfrak M_k$ vs. model $\mathfrak M_\ell$, can then be estimated as $\frac{\widehat \pi(\mathfrak M_k|y)}{\widehat \pi(\mathfrak M_\ell|y)}\times\frac{p_\ell}{p_k}$.

Note that the $f_k(y | \theta_s)$s are usually already computed for all models $k$ and at each iteration $s$ during the posterior sampling process, for instance to compute the acceptance rate of the above-described Metropolis-Hastings algorithm.

\paragraph{Model-specific inference.}

Moreover, the posterior sample $(\theta_1,\ldots, \theta_S)$ can be used as proposals in an importance sampling approach to perform inference within each candidate model $\mathfrak M_k$, with importance weights equal to: 
$$
w_k(\theta_s) = \frac{p_k f_k(y | \theta_s)} { \sum_j p_j f_j(y | \theta_s)}.
$$

Hence, an attractive feature of this approach is that BMA, Bayesian inference and posterior probabilities for {\em all models} are obtained using a single MCMC sampling algorithm, instead of requiring a new MCMC run for each model. 

The quality of the weighted particle system $(\theta_s, w_k(\theta_s))_{1\leq s\leq S}$ as an approximation of the posterior distribution for model $\mathfrak M_k$ is conveniently measured by the effective sample size (ESS) \cite{Doucet01}:
\begin{eqnarray}
ESS_k &=& \frac{\left(\sum_s w_k(\theta_s)\right)^2}{\sum_s w_k(\theta_s)^2}.
\end{eqnarray}
$ESS_k$ takes values between $1$ and $S$ 
and is interpreted as the size of an iid posterior sample carrying the same amount of information as the weighted particle system $(\theta_s, w_k(\theta_s))_{1\leq s\leq S}$. 

The following tight bounds on the estimation variance of $\widehat \pi(\mathfrak M_k|y)$ and the $ESS_k$ show that these are well-behaved numerically:
\begin{proposition}\label{prop:bounds}
$ $ \newline
\begin{enumerate}
\item The variance of the posterior probability estimate (\ref{eq:posterior_probability}) is bounded above by:
$$
\mathbb V \left[ \widehat \pi(\mathfrak M_k|y) \right]  \leq \frac{1}{S} \pi(\mathfrak M_k|y)(1 - \pi(\mathfrak M_k|y))
$$
\item $ESS_k $ is bounded below by:
$$
ESS_k \geq S \times\widehat \pi(\mathfrak M_k|y)
$$
\end{enumerate}
\end{proposition}

\begin{proof}
Observe that the $(w_k(\theta_s))_{1\leq S}$ are iid, between $0$ and $1$, with mean:
$$
\mathbb E[w_k(\theta_s)] = \int \frac{p_k f_k(y | \theta)} { \sum_j p_j f_j(y | \theta)} \pi(\theta|y) = \pi(\mathfrak M_k|y).
$$
Since $0\leq w_k(\theta_s)^2 \leq w_k(\theta_s)\leq 1,$ we have $\mathbb E[w_k ^2(\theta_s)] \leq \mathbb E[w_k(\theta_s)] = \pi(\mathfrak M_k|y),$ implying:
$$
\mathbb V[w_k(\theta_s)] \leq \pi(\mathfrak M_k|y) (1 - \pi(\mathfrak M_k|y)),
$$
whence the first result, noting that: 
$$
 \widehat \pi(\mathfrak M_k|y) = \frac{1}{S} \sum_{s=1}^S w_k(\theta_s).
$$
\newline
Likewise, since $w_k(\theta_s)^2 \leq w_k(\theta_s),$
$$
\sum_s w_k(\theta_s)^2 \leq \sum_s w_k(\theta_s),
$$
which implies the second result.
\end{proof}

The lower bound on the ESS for model-specific inference highlights the fact that the mixture model approach to BMA is in a sense `self-pruning', since poor ESS values are systematically associated with models which have low posterior probabilities, hence which can be left out of the analysis entirely. On the other hand, it is perfectly possible to have high ESS values for a model with low posterior probability, provided that the support of its posterior density overlaps that of another, high posterior probability, model.

%
%
%

\section{Illustration: Poisson / Geometric selection}
\label{sec:Illustration}

We exploit Example 3.1. from \cite{Kamari2014}: 
\newline

Consider a model choice test between a Poisson  $\mathcal P(\lambda)$ and a Geometric $\mathcal Geo(p)$ distribution, where the latter is defined as a number of failures and hence also starts at zero. We can represent the mixture~(\ref{eq:mixture}) using the same parameter  $\theta := \lambda$ in the two distributions if we set  $p = 1/(1 + \lambda).$ The resulting mixture is then defined as
$$
\frac{1}{2}\left\{\mathcal P_\lambda^n + \mathcal Geo_{1/(1+\lambda)}^n\right\},
$$
with equal prior weights for each model. This common parameterization allows the use of Jeffreys' (1939) improper prior $\pi(\lambda) = 1/\lambda$, since the resulting posterior is then proper. 

The likelihood and priors for a realization $y=(y_1,\ldots,y_n)$ of the above mixture are then given by:
\begin{eqnarray*}
f(y | \lambda) &=& \frac{1}{2}\left\{d\mathcal P^n_\lambda(y) + d\mathcal Geo_{1/(1+\lambda)}^n(y)\right\}\nonumber\\
&=& \frac{1}{2}\left\{\prod_k d\mathcal P_\lambda(y_k) + \prod_k d\mathcal Geo_{1/(1+\lambda)}(y_k)\right\}\\
&=& \frac{1}{2}\left\{\frac{e^{-\lambda n} \lambda^{S_n}}{\prod_k y_k!}  + \frac{\lambda^{S_n}}{(1+\lambda)^{S_n+n}}\right\}\\
\pi(\lambda) &\propto& 1/\lambda
\end{eqnarray*}
where $S_n := \sum_k y_k$.

Multiplying by $\pi(\lambda)$ yields the posterior marginal distribution of $\lambda$, which is proportional to the sum of the unnormalized posterior distributions within each model:
\begin{eqnarray*}
\pi(\lambda|y) &\propto& \frac{e^{-\lambda n} \lambda^{S_n-1}}{\prod_k y_k!}  + \frac{\lambda^{S_n-1}}{(1 + \lambda)^{S_n+n}}.
\end{eqnarray*}
Integrating each term with respect to $\lambda$ then gives the marginal distribution in the Poisson and Geometric models: 
$$
m_0 = \frac{\Gamma(S_n)}{n^{S_n}\prod_k y_k!};\quad m_1 = \frac{\Gamma(S_n)\Gamma(n)}{\Gamma(S_n+n)}
$$
whence the Bayes factor reads:
$$
B_{01}(x) = \frac{\Gamma(S_n+n)}{n^{S_n}\prod_k y_k!\Gamma(n)}.
$$
Note that it is not well-defined mathematically since the prior on  $\lambda$ is improper. In the mixture-model setting on the contrary, there is no difficulty in using such an improper prior, as long as the posterior $\pi(\lambda|y)$ is proper. 

$\lambda$'s posterior distribution can be sampled using a standard random-walk Metropolis-Hastings algorithm, initialized for instance using the posterior mean or mode within each models: $\hat\lambda = (S_n-1)/n$ for the Poisson and $\hat\lambda = (S_n-1)/(n+1)$ for the Geometric model. In practice, we used a Gaussian random walk, tuning the scale-parameter in order to attain an acceptance rate in $[.2, .8]$.

\paragraph{Results}

\begin{figure}
\includegraphics[width=.5\textwidth]{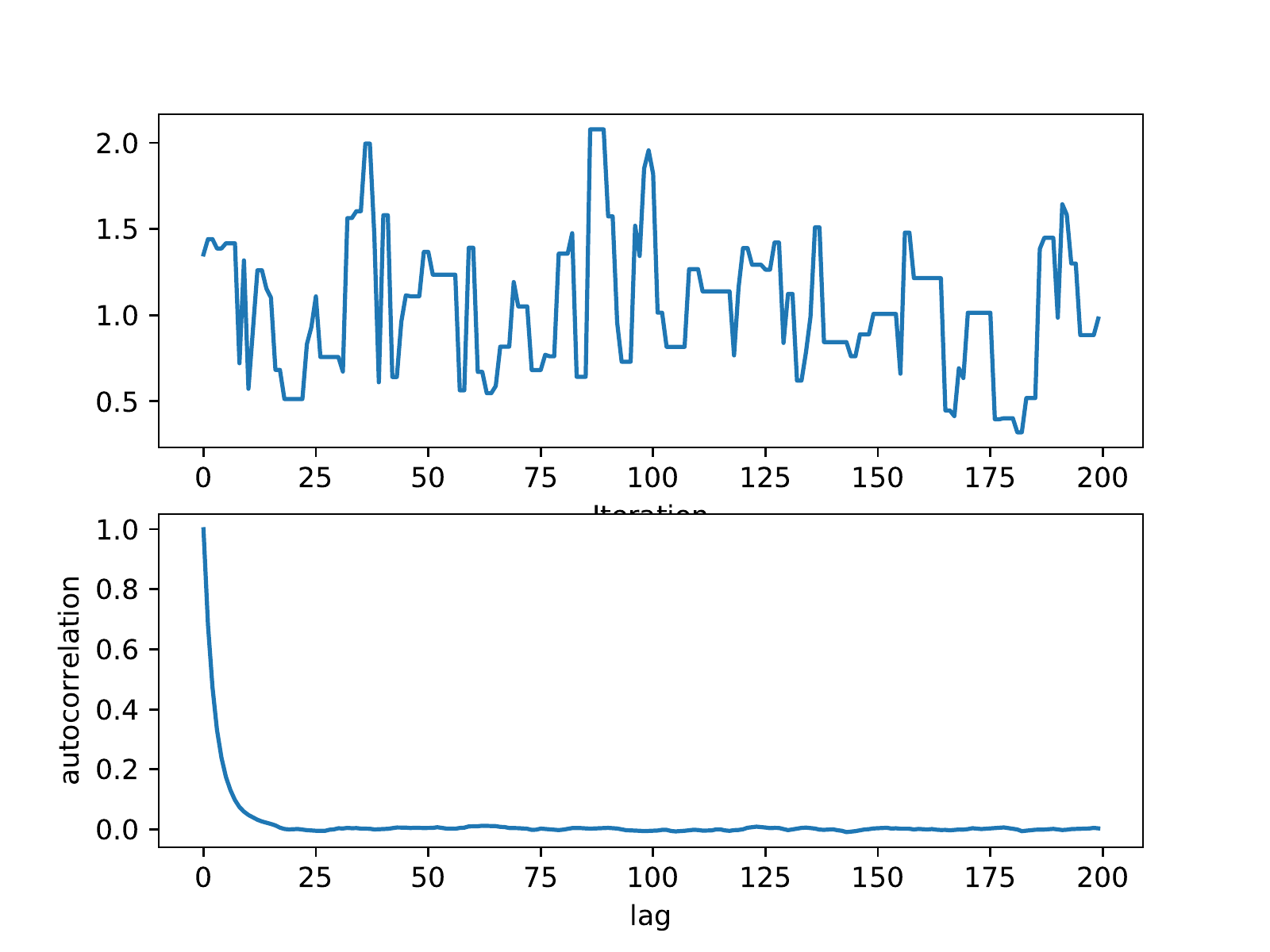}
\includegraphics[width=.5\textwidth]{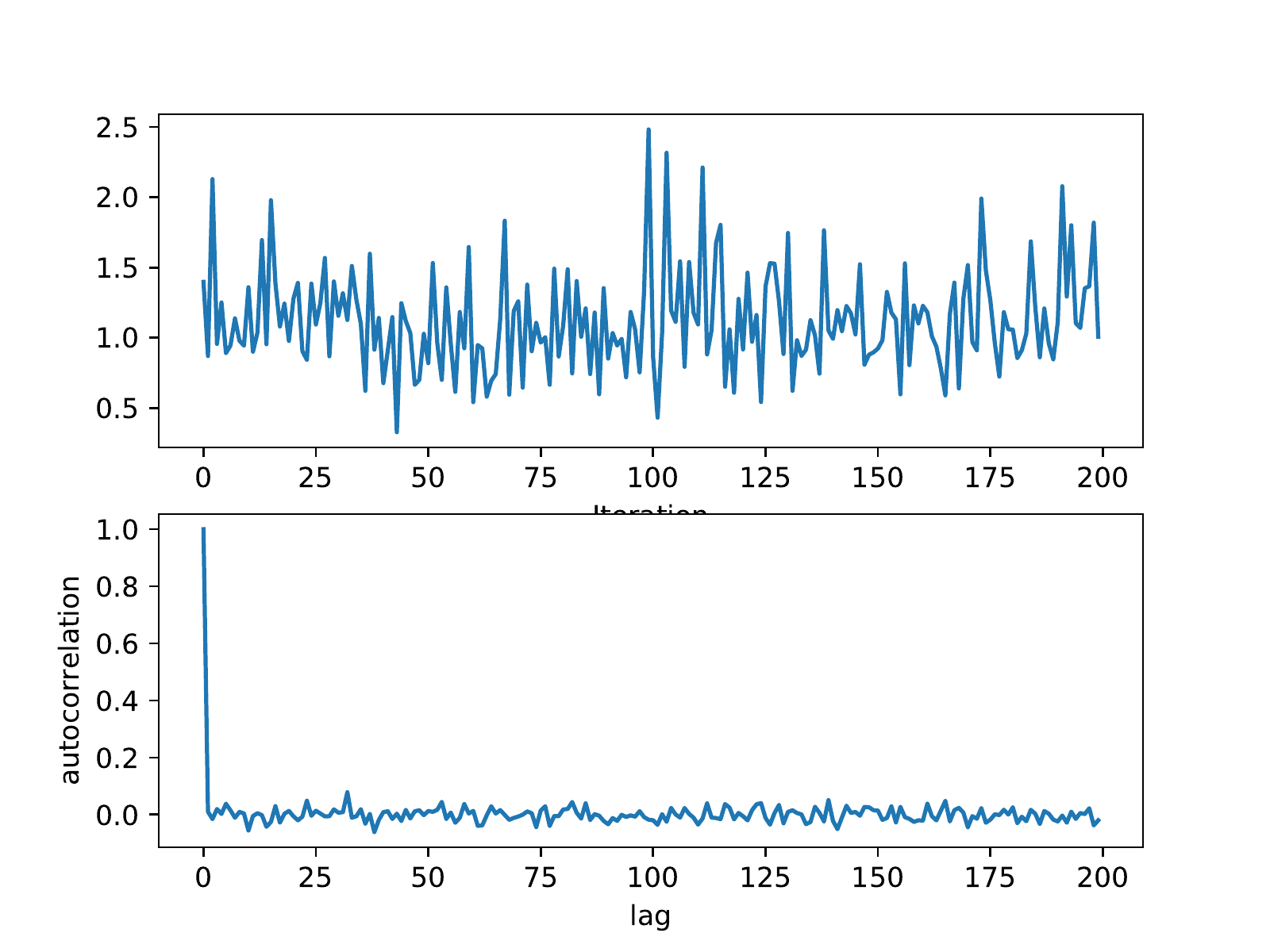}
\caption{\label{fig:cv_plot} MCMC for $\lambda$ (top) and autocorrelation (bottom). Left: before thinning, right: after thinning (1/50 iterations kept)}
\end{figure}

We simulated $n=10$ Poisson variates with $\lambda=1.$ Then, we ran $10^5$ iterations of the above Metropolis-Hastings algorithm, which took about ten seconds on a 2.4 GHz Mac Book Pro. The acceptance rate was close to $54\%$. Figure~\ref{fig:cv_plot} shows a convergence plot of the MCMC used to sample from $\lambda$'s posterior distribution. Mild autocorrelation was suppressed by thinning. This allowed to consider the output as iid, which enabled computing confidence intervals for the Bayes factor Monte-Carlo estimate using the central limit theorem. Results for the computation of the Bayes factor are presented in Table~\ref{tab:BF}. These can be used for instance to select the `most probable' model for the data, here the true one $\mathfrak M_0$, since the estimated Bayes factor is greater than one.

\begin{table}
\centering
\begin{tabular}{cc}
$BF_{01}$& 8.38 \\
$\widehat{BF}_{01}$& 8.41 \\
LCL& 8.15 \\
UCL& 8.69 \\
\end{tabular}

\caption{\label{tab:BF} Computation of the Bayes factor}
\end{table}

\begin{figure}
\begin{center}
\includegraphics[width=.5\textwidth]{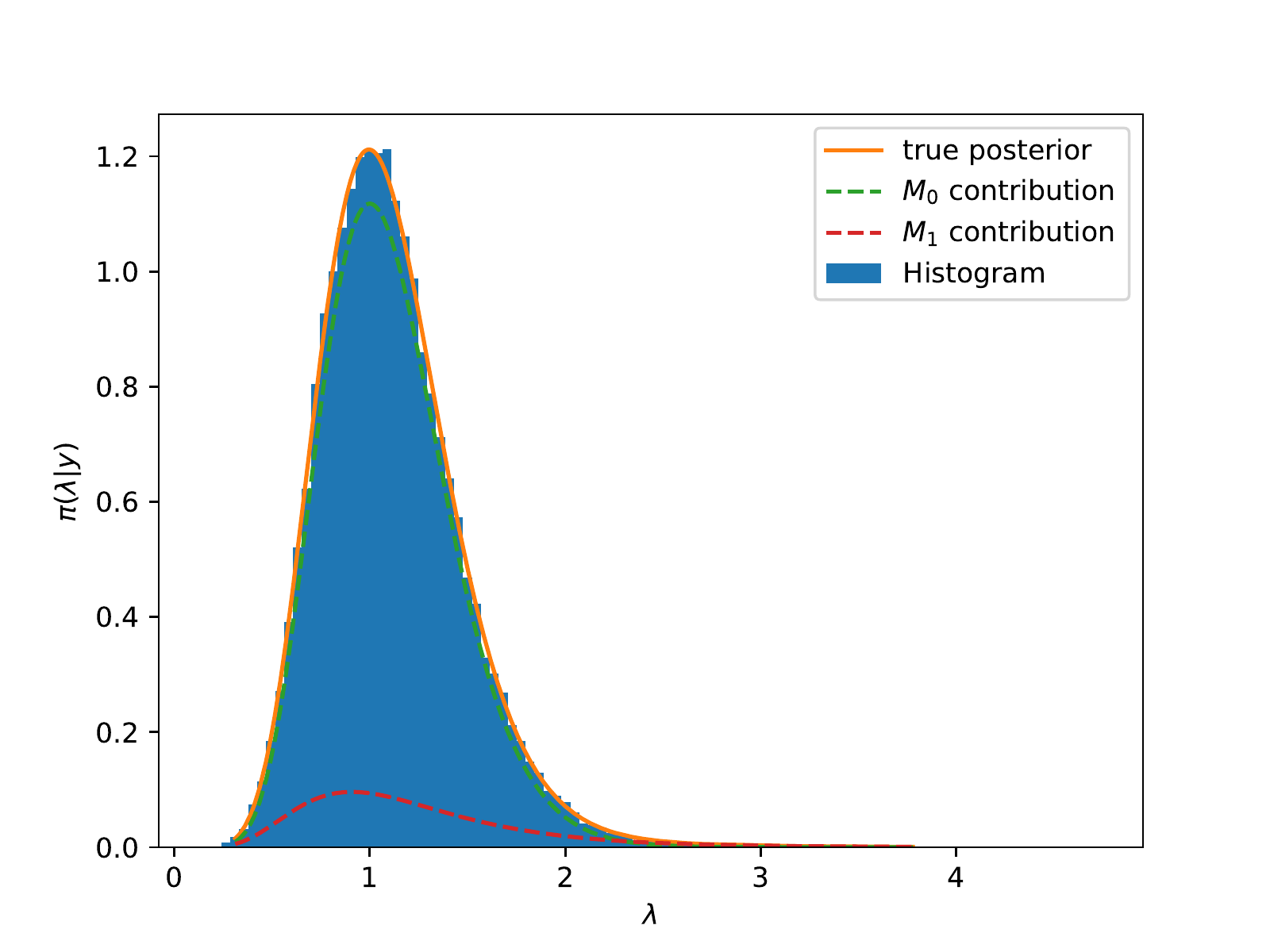}\\
\includegraphics[width=.5\textwidth]{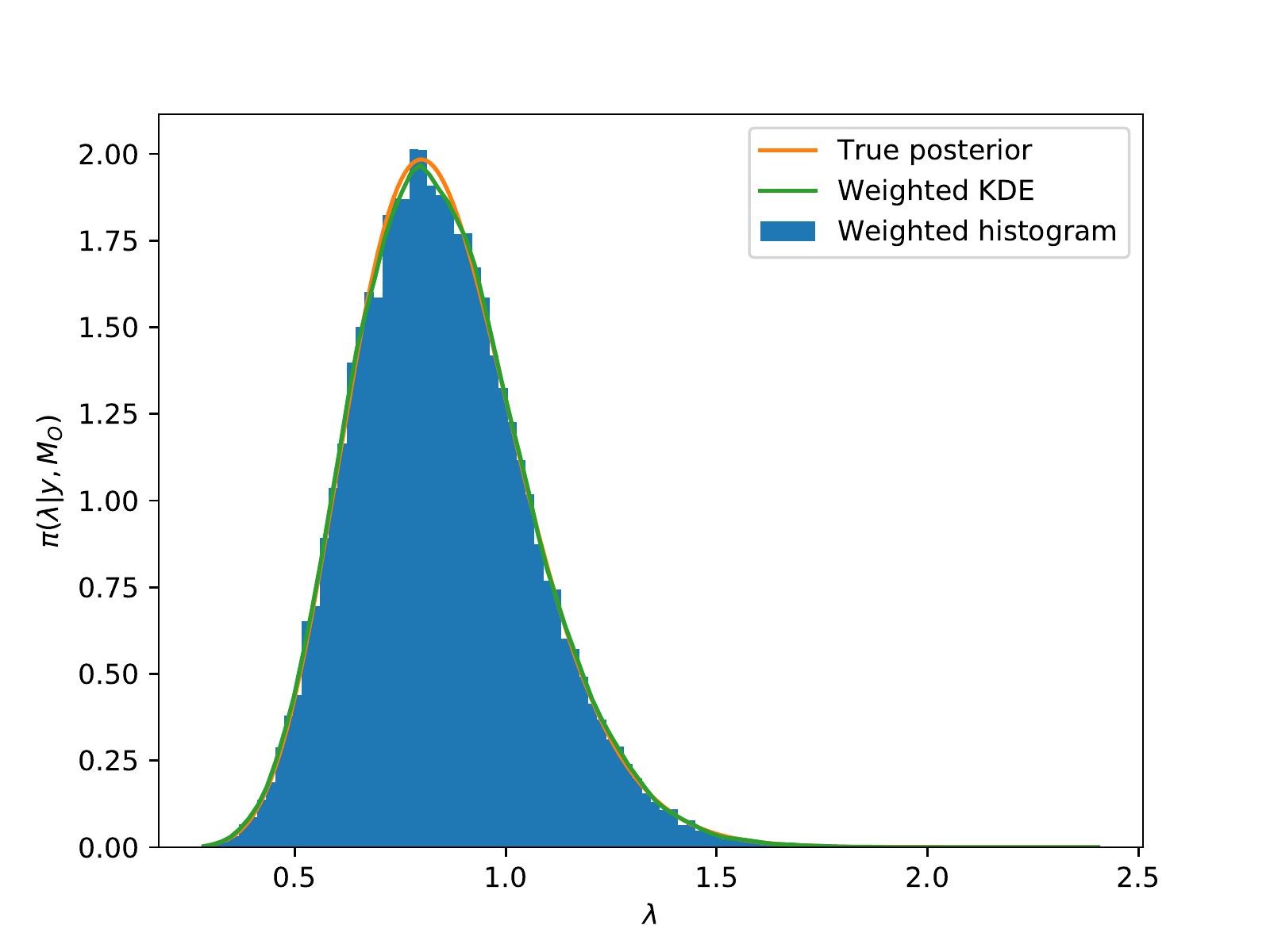}\hfill
\includegraphics[width=.5\textwidth]{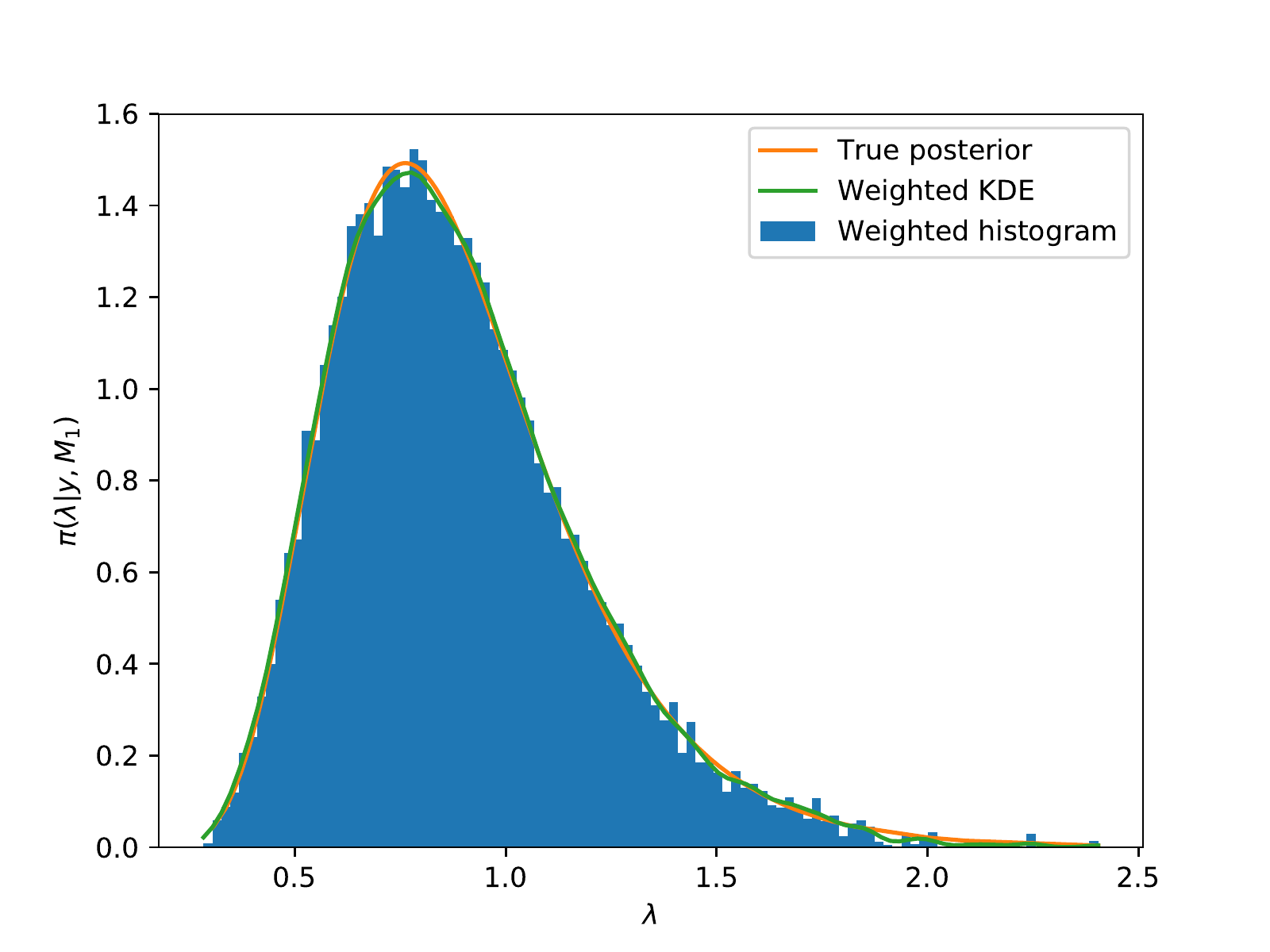}
\end{center}
\caption{\label{fig:BMA} Histogram of model-averaged posterior density for $\lambda$ (top), histogram of posterior density for $\lambda$ in model $\mathfrak M_0$ (bottom left) and in model $\mathfrak M_1$ (bottom right).}
\end{figure}

To illustrate the fact that both model-specific and model-averaged inference can be performed using this approach, Figure~\ref{fig:BMA} shows the histogram of the Bayesian model-averaged posterior distribution for $\lambda$, as well as the (weighted) histograms for the posterior distributions within each candidate model, obtained through importance sampling. The ESS associated to the importance weights, equal to $99\,396$ and $69\,803$ for models $\mathfrak M_0$ and $\mathfrak M_1$, respectively, were well above the lower bounds in Proposition~\ref{prop:bounds}, equal to: $71\,027$ and $28\,972$, respectively, due to the large overlap of the posterior densities.

Using these weighted samples, we can for instance estimate $\lambda$, which in both models is equal to the data's expectation: $\lambda = \mathbb E[y].$ The ensuing results are given in Table~\ref{tab:lambda}. Results using BMA are very close to those obtained for the true model, due to its high posterior probability ($89\%$).

\begin{table}
\centering
\begin{tabular}{c|ccc}
Model & $\lambda$ estimate & LCL & UCL \\
$\mathfrak M_0$& 1.1& 0.55& 1.84 \\
$\mathfrak M_1$& 1.21& 0.45& 2.58 \\
$\mathfrak M_0+\mathfrak M_1$& 1.11& 0.54& 1.92 \\
\hline
\hline
True value & 1.0 & & 
\end{tabular}

\caption{\label{tab:lambda} Estimation of $\lambda$}
\end{table}


\section{Illustration: Linear code validation}

Consider the following formulation for linear code validation, where we wish to explain experimental data $\bs y := (y_1, \ldots, y_n)$ from explanatory variables $\bs x := (x_1, \ldots, x_n)$, using a linear code $h(\bs x) \theta:$
    \begin{eqnarray}\label{eq:M0}
    \mathfrak M_0: \bs y &=& h(\bs x) \theta + \bs\varepsilon,
    \end{eqnarray}
    or, accounting for a possible model bias $\bs\delta$ \cite{KOH2001},
    \begin{eqnarray}\label{eq:M1}
    \mathfrak M_1: \bs y &=& h(\bs x) \theta + \delta(\bs x) + \bs \varepsilon,
    \end{eqnarray}
where:
\begin{eqnarray*}
\bs\varepsilon|\kappa,\sigma^2 &\sim& \mathcal N(0,\kappa \sigma^2)\\
\delta(\cdot)|\sigma^2 &\sim& GP( m(\cdot), \sigma^2 K(\cdot, \cdot) )\\
\pi(\theta) &\propto & 1 \\
\pi(\sigma^2) &\propto & \sigma^{-2} \\
\pi(\kappa) &\sim& \mathcal U([0,1])
\end{eqnarray*}

For this exercise, the correlation function $K$ is chosen as a squared exponential covariance function: $K(x,x')=e^{-(\frac{x-x'}{\gamma})^2}$, with correlation length $\gamma=0.2.$ For simplicity, this assumed to be known in advance, even though such is usually not the case in real-life practice. 

$k=1/\kappa$ represents a (squared) signal-to-noise ratio between the discrepancy and the measurement errors \cite{Damblin2016}. Hence, we assume {\em a priori} that $k>1$, otherwise the discrepancy is swamped in the observation noise, and becomes very difficult to detect. This is why we choose to model $\kappa = 1/k$ as uniformly distributed between $0$ and $1$, though a more general Beta distribution could be used.

Meanwhile, because $\lambda$ and $\theta$ are common to both models, they can be endowed improper priors, as long as they give rise to a proper posterior distribution. A standard choice for Gaussian linear regression models is Jeffreys' prior:

\begin{eqnarray}\label{eq:prior_theta_lambda}
\pi(\theta)\pi(\lambda) &\propto& 1/\lambda.
\end{eqnarray}

In Appendix~\ref{app:propriety}, we show that this rather natural choice indeed induces a proper posterior distribution. This is in accordance with \cite{Damblin2016}, who dealt with a similar prior choice in the context of Bayes factor model selection, using {\em intrinsic} Bayes factors \cite{Berger96} to overcome the ill-defined Bayes factor. Hence, the present illustration shows how the same conclusions can be reached, using a much simpler approach, and without any foundational difficulties concerning the definition of the Bayes factor.

Finally, we choose not to favor one model over the other {\em a priori}, hence we set their prior weights equal to: $p_0 = p_1 = 1/2.$

Figure~\ref{fig:linear_data} shows how the dataset $y$ was generated for this numerical experiment, following model $\mathfrak{M}_1.$ Upon applying our methodology for Bayesian model averaging, we therefore expect to find $\pi(\mathfrak M_0|y) < \frac{1}{2}.$

\begin{figure}\label{fig:linear_data}
\centering
\includegraphics[width=0.6\textwidth]{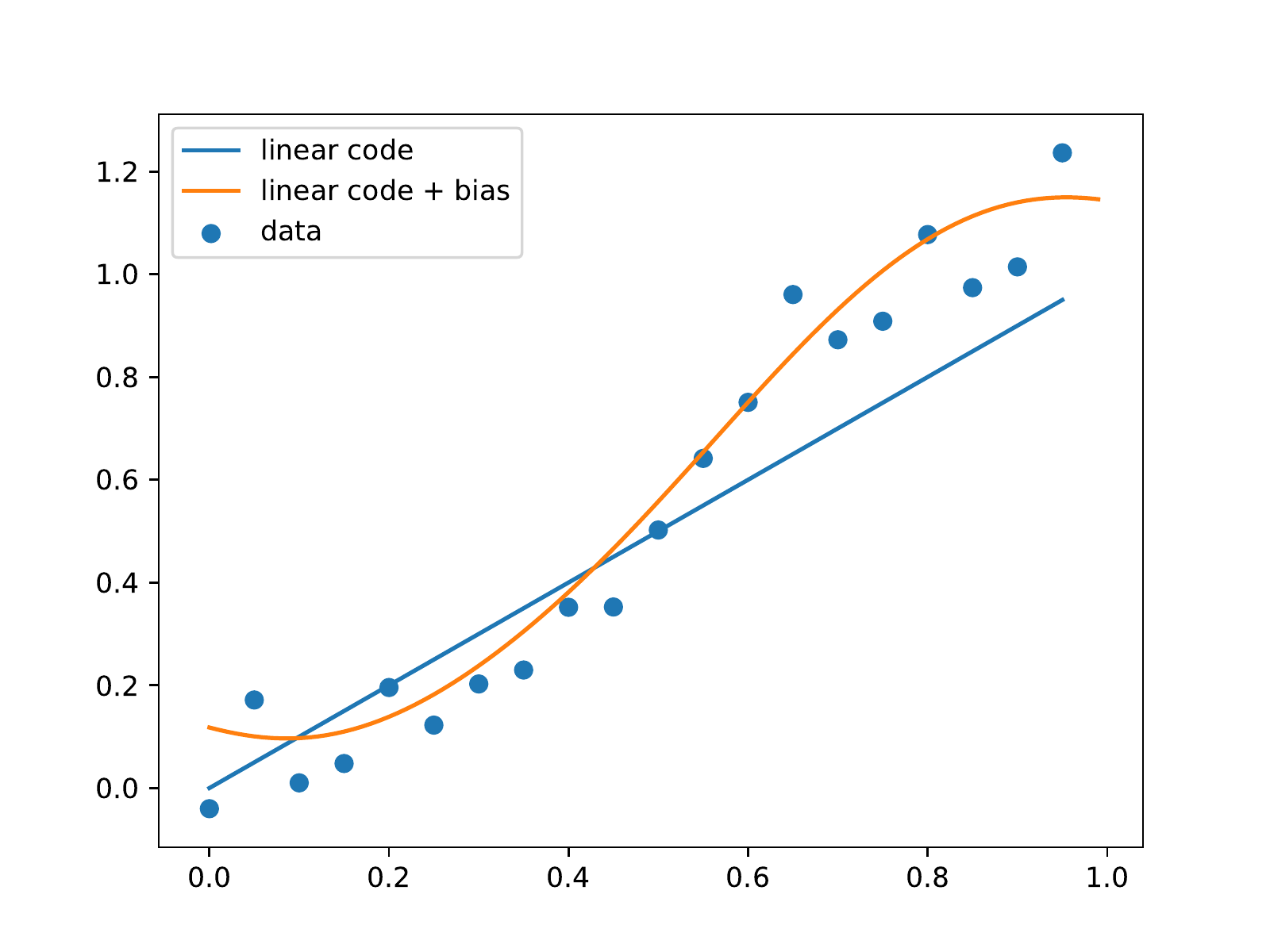}
\caption{Simulated data: linear tendency (blue line), linear tendency plus Gaussian process discrepancy (orange curve) and observations (dots).}
\end{figure}

\paragraph{Implentation details.} Appendix~\ref{app:propriety} shows that uncertain parameters $\theta$, $\lambda$ and $\delta$ can be integrated out analytically, leaving $\kappa$'s marginal posterior distribution alone to be sampled. This was done using ans independent Metropolis-Hastings (IMH) sampler, based on the uniform prior, meaning that in the algorithm described in Section~\ref{sec:algo}, we chose the proposal kernel for $\kappa$ equal to: $q(\kappa|\kappa^\ast)=\bs 1_{[0,1]}(\kappa)$.

\subsection{Results}

We ran $10\,000$ iterations of the IMH, and then thinned the resulting output, yielding approximately independent draws $(\kappa_s)_{1\leq s\leq S}$ from the BMA posterior. From these, conditional posterior probabilities $\pi(\mathfrak M_0 | \kappa_s, y)$ for model $\mathfrak M_0$ where computed, using Equation~(\ref{eq:posterior_probability}). Averaging these yielded an estimate of the (unconditional) posterior probability $\pi(\mathfrak M_0 | y)$. As shown in Table~\ref{tab:posterior_probability}, this probability was significantly below $50\%$, meaning that the data contained enough information to detect a non-zero discrepancy.

\begin{table}
\centering
\begin{tabular}{c|c|c}
$\widehat \pi(\mathfrak M_0|y)$ & $95\%$LCL & $95\%$UCL\\
\hline
28.2\% & 28.4\% & 28.7\%
\end{tabular}
\caption{\label{tab:posterior_probability} Monte-Carlo estimate of the posterior probability that the data was simulated under model $\mathfrak M_0$ (without discrepancy), together with confidence bounds. }
\end{table}

Next, we reconstructed a BMA posterior sample $(\theta_s, \lambda_s, \delta_s)_{1\leq s\leq S}$ for the remaining parameters using the following steps:
\begin{enumerate}
\item for $s=1,\ldots,S$, draw $\zeta_s \sim \mathcal B(1,\pi(\mathfrak M_1 | \kappa_s, y))$
\item simulate $(\theta_s, \lambda_s, \delta_s)$ from there full posterior conditional distribution under $\mathfrak M_0$ if $\zeta_s=0$, and under $\mathfrak M_1$ otherwise.
\end{enumerate}

\begin{figure}
\includegraphics[width=0.5\textwidth]{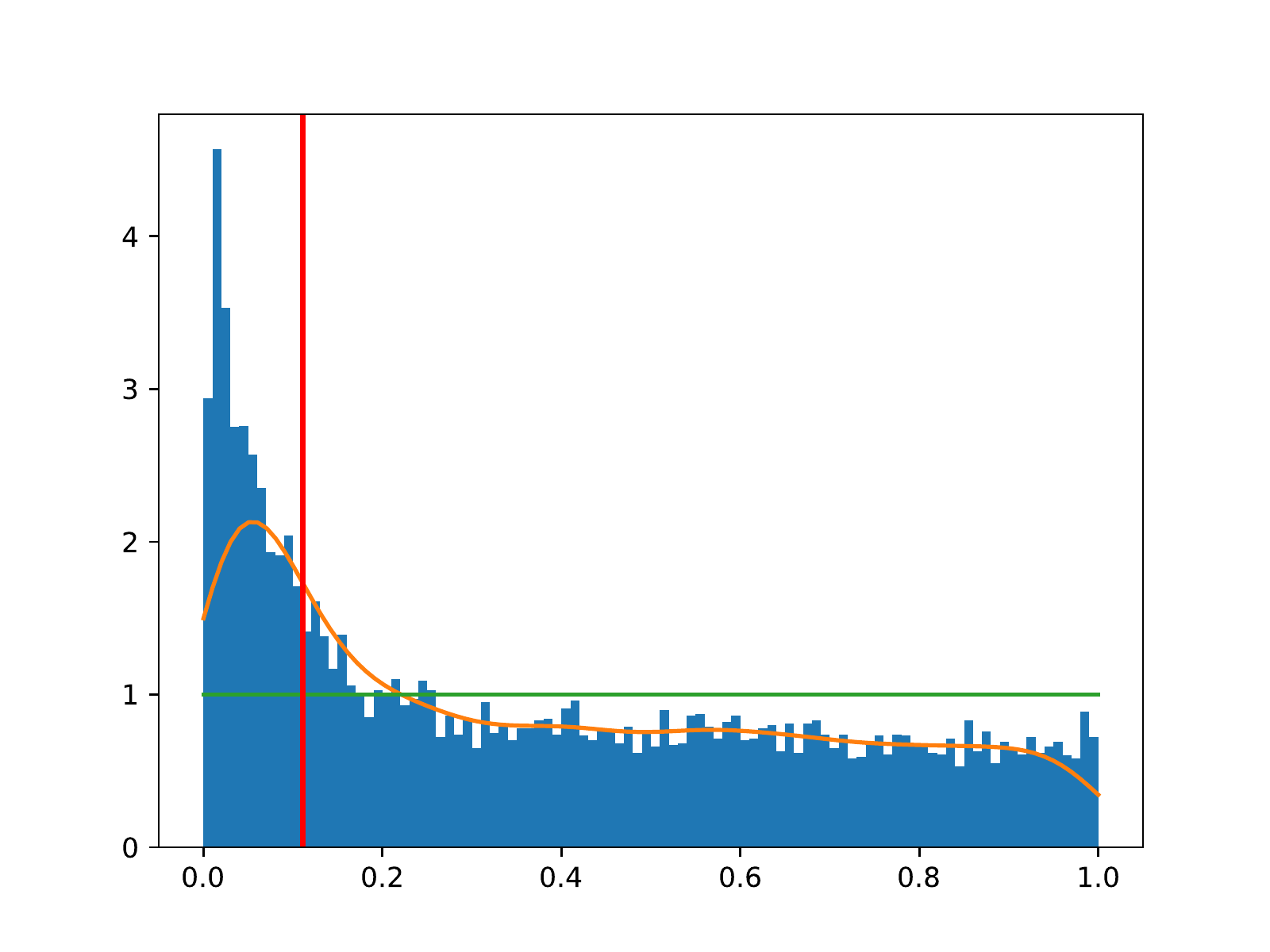}
\includegraphics[width=0.5\textwidth]{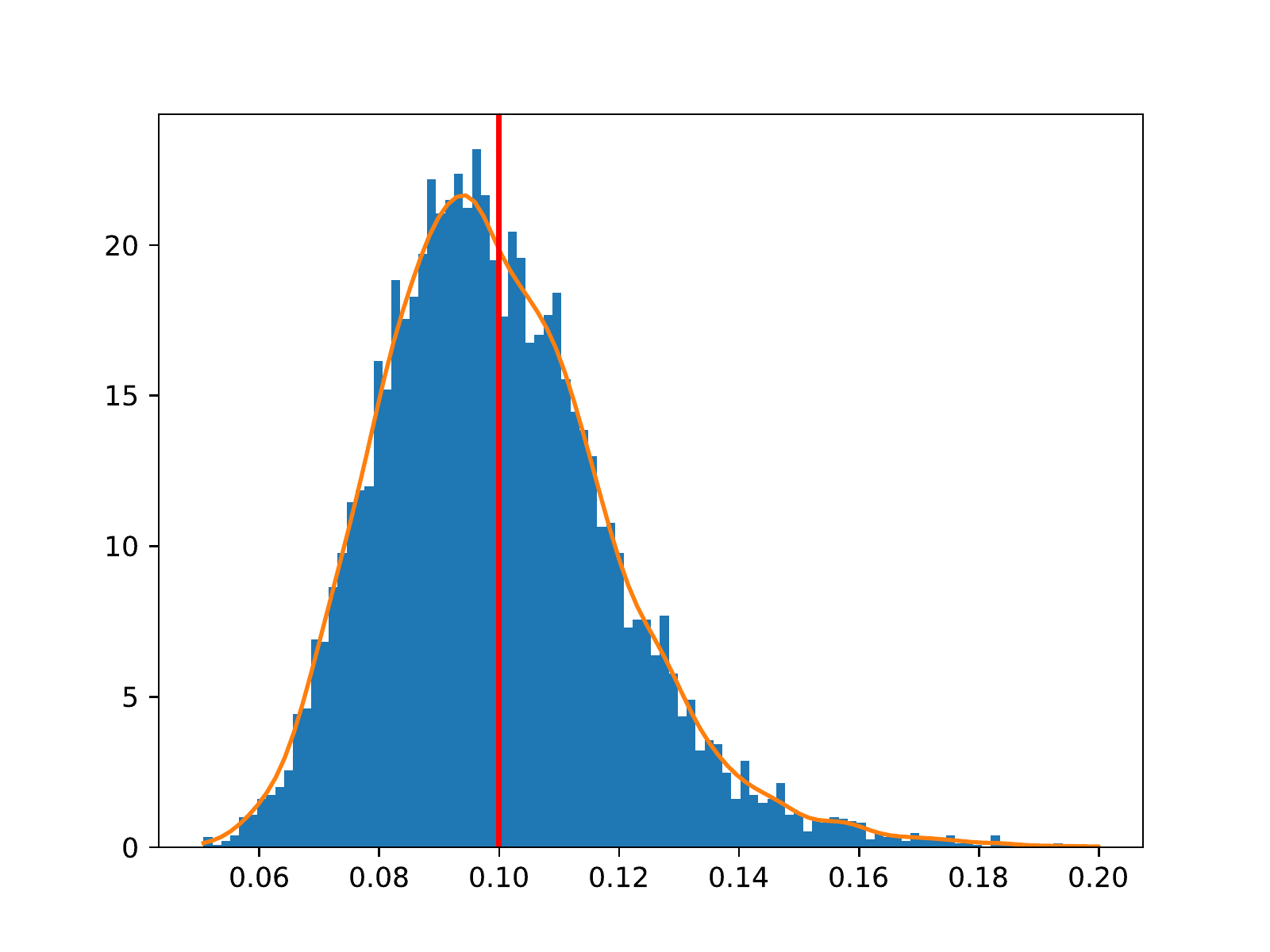}\\
\caption{\label{fig:posterior_distribution}
Model-averaged posterior densities (bars and orange curves) of $\kappa$ (left) and $\lambda$ (right), vs. prior densities (green curves)and real values (vertical lines).
}
\end{figure}

Figure~\ref{fig:posterior_distribution}\ shows the histograms of the BMA marginal posterior distributions for both $\kappa$ and $\lambda^2$. We can verify visually that $\kappa's$ BMA posterior distribution is a mixture of the uniform prior, which is the posterior under model~$\mathfrak M_0$, since then $\kappa$ is independent from the data, and the posterior distribution under~$\mathfrak M_1$, which is more or less peaked around the real value. Likewise, the observation variance $\lambda^2$ is seen to be relatively well estimated.

Finally, Figure~\ref{fig:predictions} shows the BMA prediction of the central tendency, as well as in models $\mathfrak M_0$ and $\mathfrak M_1$. The latter were directly derived from the reconstructed posterior sample, meaning that we didn't perform importance sampling as described in Section~\ref{sec:algo}, but simply assigned each BMA posterior draw $(\kappa_s, \theta_s, \delta_s, \lambda_s)$ to model $\mathfrak M_0$ if $\zeta_s=0$, and to model $\mathfrak M_1$ if $\zeta_s = 1$. As could be expected, the BMA posterior estimate is almost as accurate as the posterior estimate under the true model (here $\mathfrak M_1$), since the algorithm correctly identified the latter as being most probable.

\begin{figure}
\centering
\includegraphics[width=0.5\textwidth]{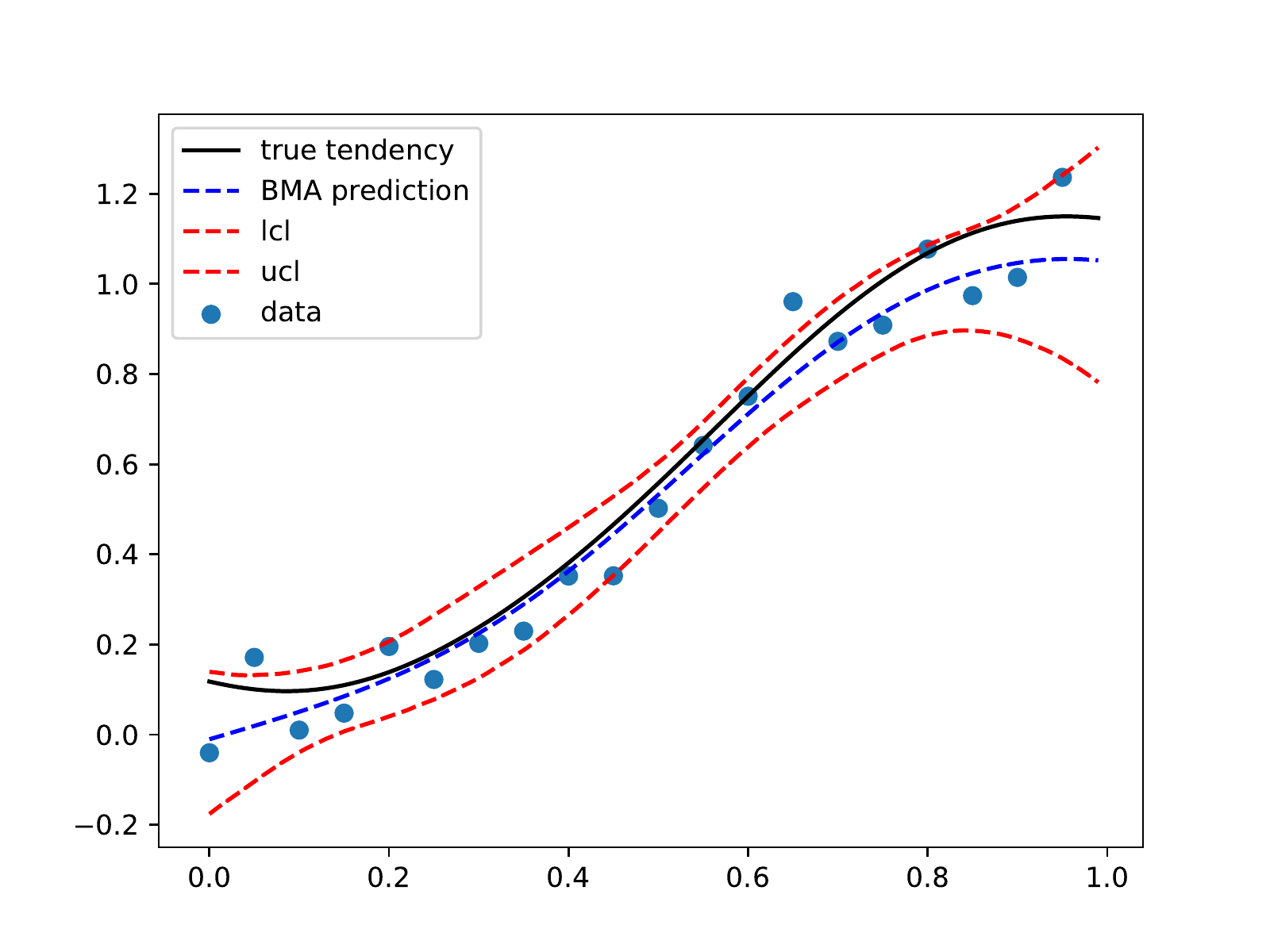}\\
\includegraphics[width=0.5\textwidth]{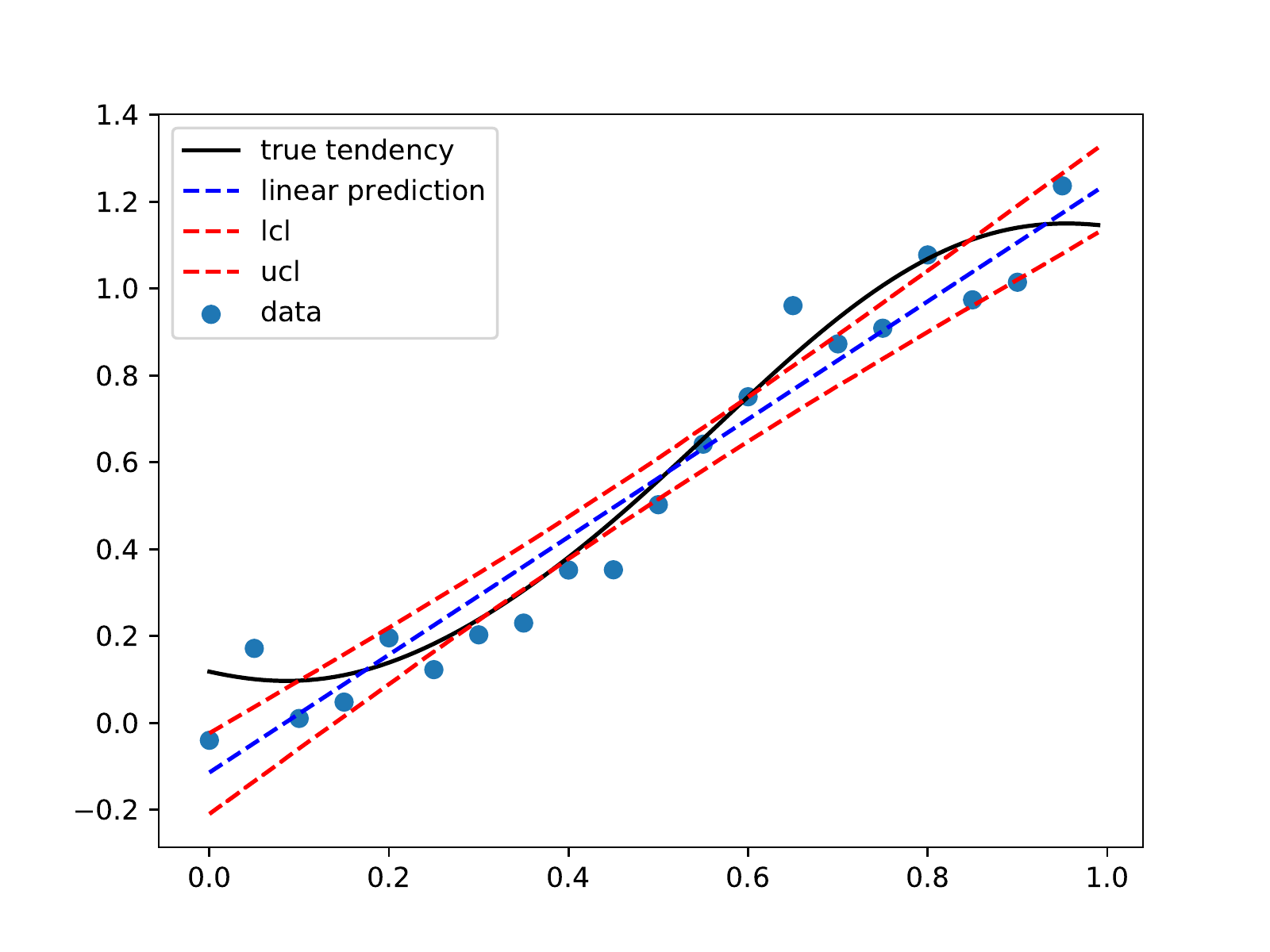}\hfill
\includegraphics[width=0.5\textwidth]{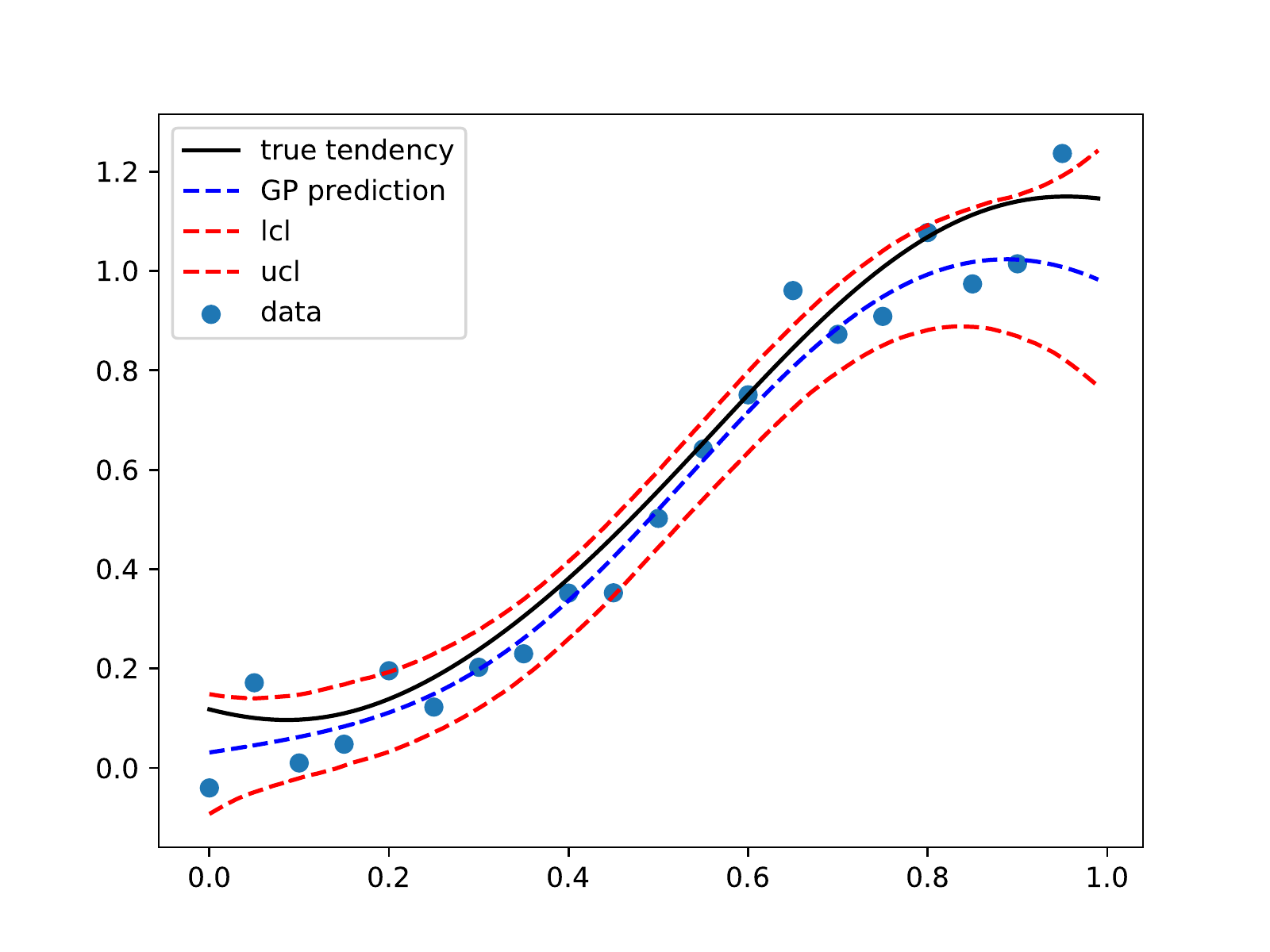}\\
\caption{\label{fig:predictions}
Estimate of the tendency under BMA (top), $\mathfrak M_0$ (bottom left)
and $\mathfrak M_1$ (bottom right).
}
\end{figure}

\section{Discussion}\label{sec:discussion}

In this paper, we have shown that the mixture modeling approach to model choice problems introduced in \cite{Kamari2014}, contains as a special case the traditional Bayesian model selection and averaging paradigm, as described for instance in \cite{Hoeting99}. This result leads to a re-formulation for the BMA posterior distribution, making it easy to sample directly using standard MCMC algorithms. This also naturally allows improper priors for shared parameters, a key feature of the approach advocated by \cite{Kamari2014}.

Moreover, once the BMA posterior is sampled, very simple estimates for the posterior probabilities of all candidate models are available, as well as importance sampling estimates for their posterior distributions. Tight bounds on the posterior probabilities estimation variances, and the ESSs of the model-specific importance weights show that these estimates are well-behaved, and exhibit a `self-pruning' property: the lower bound on the ESS associated with each model is directly proportional to its estimated posterior probability. This guarantees that high posterior-probability models are well estimated.

The main difficulty of the advocated approach consists in sampling from the encompassing mixture model posterior, which may be challenging if many candidate models are present, especially if the total vector parameter is high-dimensional, or if the resulting posterior is multi-modal. Hence, some kind of pre-treatment may still be needed in practice when the list of candidate model is large, such as pre-selecting some models using a cheap approximation to the Bayes factor, as is done in \cite{Hoeting99}. Another workaround is to divide the set of candidate models into overlaping subsets, estimate a mixture model on each subset, then reconstruct the complete BMA posterior.

Furthermore, a crucial point that remains to be investigated is the possibility of sharing parameters across models, which has the double effect of reducing the parameter space dimension, and allowing improper priors to be used. The most favorable setting seems to be variable selection, {\em i.e.} when all candidate models are contained in an encompassing model containing all the parameters. In this case, if $d$ is the number of parameters considered, the reduction in dimension space from $\sum_{k=1}^d k \times C_d^k = d\times2^{d-1}$ to $d$ can be dramatic. Likewise, in theory the total number of parameters could always be reduced to the maximum number of parameters within one given model, via systematic re-parameterizations. However, it is not clear for a given problem what parameters should be shared, or what common priors should be chosen.

\par

\small


%
%
%
%
%

\appendix

\section{Posterior propriety}\label{app:propriety}

Due to the use of improper priors on $\lambda$ and possibly on $\theta$, it is necessary to verify that the posterior distribution in the mixture model (\ref{eq:mixture_model_posterior}) is proper. In our case, this simply boils down to showing that the posterior distribution from both models is proper.


We start with propriety under model $\mathfrak M_0$. This is obtained by showing that:
$$
m_0(y) = \int_{\theta,\lambda} f_0(y|\theta,\lambda,x)\pi(\theta)\pi(\lambda) d\theta d\lambda 
$$
is finite; note that $\delta$ and $k$ need not be considered here, since they do not intervene in the calculation of $f_0$, and the parameters are pairwise independent. Following \cite{Kamary17}, we first integrate $\theta$ out. In Gelfand and Smith's bracket notation:
\begin{eqnarray*}
[ y | \theta, \lambda, \mathfrak M_0, x ] [ \theta ] [ \lambda]
&=& \frac{\lambda^{-n-1}}{\sqrt{2\pi}^n}  \exp\left\{ -\frac{1}{2\lambda^2}|| y - g(x)\theta ||^2\right\},
\end{eqnarray*}

Next, to integrate $\theta$ out, we develop the sum of squares, simplifying the notation $g(x)$ to $g_x$:
\begin{eqnarray*}
|| y - g_x\theta ||^2 
&=& 
(y - g_x\theta)^\top (y - g_x\theta)\\
&=& y^\top\! y + \theta^\top\! {g_x}^\top\! g_x \theta - 2y^\top\! g_x \theta
\end{eqnarray*}
Identifying the above expression with the development of $(\theta - \widehat\mu_0)^\top \widehat\Sigma_0^{-1}(\theta - \widehat\mu_0)$, we find that the conditional posterior distribution of $\theta$ given $\lambda$, is proper, and given by:
\begin{eqnarray*}
\theta |y, \lambda, \mathfrak M_0, x &\sim& \mathcal N(\widehat\mu_0, \widehat\Sigma_0),
\end{eqnarray*}
where:
$$
\widehat\mu_0 = (g_x^\top\, {g_x})^{-1} {g_x}^\top y;\quad \widehat\Sigma_0 = ({g_x}^\top g_x)^{-1}.
$$
By integrating $\theta$ out, we obtain the product of the partially integrated likelihood by the prior for $\lambda$:
\begin{eqnarray}
[ y | \lambda, \mathfrak M_0, x ] [ \lambda] 
&=& \frac{\lambda^{-(n-p)-1}}{\sqrt{2\pi}^{n-p}} \sqrt{|\widehat\Sigma_0|}\exp\left\{ -\frac{1}{2\lambda^2}y^\top(I_n - g_x(g_x^\top\!g_x)^{-1} g_x^\top)y\right\}.
\label{eq:M0_lambda_marginal}
\end{eqnarray}
A key point is to show that $y^\top(I_n - g_x(g_x^\top\!g_x)^{-1} g_x^\top)y$ is strictly positive with probability one for $y$, in both models, and whatever the parameters' values. To see this, note that $g_x(g_x^\top\!g_x)^{-1} g_x^\top$ is the matrix of the orthogonal projection $\pi_{g_x}$ on $Vect(g_x)$, the subspace of $\mathbb R^n$ spanned by the columns of $g_x$. Hence, $(I_n - g_x(g_x^\top\!g_x)^{-1} g_x^\top)$ is the matrix of the orthogonal projection $\pi_{g_x^\bot}$ on the orthogonal complement of $Vect(g_x)$.
This turns out to be equal to:
\begin{eqnarray*}
y^\top(I_n - g_x(g_x^\top\!g_x)^{-1} g_x^\top)y &=& ||y-g_x \widehat\mu_0||^2.
\end{eqnarray*}
In other words, the sum of squares we wish to control is equal to zero if and only if $y$ is perfectly explained by the linear model, or, equivalently, if the observation errors $\varepsilon$ are null. But the probability that this happens is equal to zero, as long as $\varepsilon$'s variance $\lambda^2$ is strictly positive.

To integrate (\ref{eq:M0_lambda_marginal}) over $\lambda$, use the change in variables $\tau = 1/{\lambda^2}$, yielding:
\begin{eqnarray*}
[ y | \tau, \mathfrak M_0, x ] [ \tau] 
&=& \frac{1}{2}\frac{\tau^{\frac{n-p}{2}-1}}{\sqrt{2\pi}^{n-p}} \sqrt{|\widehat\Sigma_0|} \exp\left\{ -\frac{||y - g_x \widehat\mu_0||^2}{2}\tau\right\}.
\end{eqnarray*}
We recognize easily the expression of the $\mathcal Ga\left(\frac{n-p}{2}, \frac{||y - g_x \widehat\mu_0||^2}{2}\right)$ Gamma distribution, so that $m_0(y)$ is given by:
\begin{eqnarray*}
[ y | \mathfrak M_0, x ]
&=& \frac{1}{2}\frac{\Gamma(\frac{n-p}{2})}{\sqrt{\pi}^{n-p}} \sqrt{|\widehat\Sigma_0|}
	||y - g_x \widehat\mu_0||^{-(n-p)}.
\end{eqnarray*}


Propriety for model $\mathfrak M_1$ is obtained along the same lines as for $\mathfrak M_0$, except that, following still \cite{Kaniav17}, $\delta$ needs to be integrated out first, and $k$ last. Recall that we need to show that:
$$
m_1(y) = \int_{\delta,\theta,\lambda,k} f_1(y|\delta,\theta,\lambda,k,x)\pi(\delta)\pi(\theta)\pi(\lambda)\pi(k) d\delta d\theta d\lambda dk
$$
is finite. We first integrate $\delta$ out; this is immediate since 
$$
y | \theta,\lambda,k, \mathfrak M_1 \sim \mathcal N(g_x \theta; \lambda^2(I_n + k{\rm Corr})),
$$
so that 
\begin{eqnarray}
[ y | \theta, \lambda, k, \mathfrak M_1, x ]
&=& \frac{\lambda^{-n}}{\sqrt{2\pi}^n} |I_n + k{\rm Corr}|^{-1/2} \exp\left\{ -\frac{(y-g_x\theta)^\top (I_n + k{\rm Corr})^{-1}(y-g_x\theta)}{2\lambda^2} \right\}.\label{eq:delta_marginal}
\end{eqnarray}

An alternative would be to derive $\delta$'s conditional posterior. This is given by:
$$
[ y | \delta, \theta, \lambda, k, \mathfrak M_1, x ] [ \delta] \propto \lambda^{-2n}
\exp\left\{
-\frac{1}{2\lambda^2} \left[ (y - g_x \theta - \delta)^\top (y - g_x \theta - \delta) + \delta^\top (k{\rm Corr})^{-1} \delta \right]
\right\}.
$$
Here, we can develop the quadratic form into:
\begin{eqnarray*}
&(y - g_x \theta - \delta)^\top (y - g_x \theta - \delta) + \delta^\top \left(1/k{\rm Corr}^{-1} \right)\delta& \\
&=&\\
&(y\top - g_x \theta)^\top (y\top - g_x \theta) + \delta^\top (I_n + 1/k{\rm Corr}^{-1}) \delta - 2(y - g_x \theta)^\top \delta&
\end{eqnarray*}
which we identify with: $(\delta - m)^\top V^{-1} (\delta - m)$, leading to:
$$
V = (I_n + 1/k{\rm Corr}^{-1})^{-1};\quad m = (I_n + 1/k{\rm Corr}^{-1})^{-1} (y - g_x \theta).
$$
Next, we use the {\em basic marginal equality} \cite{Chib95}:
$$
[ y | \theta, \lambda, k, \mathfrak M_1, x ] = \frac{[ y | \delta, \theta, \lambda, k, \mathfrak M_1, x ][\delta]}{[ \delta | y, \theta, \lambda, k, \mathfrak M_1, x ]},
$$
which leads directly to (\ref{eq:delta_marginal}).

Next, to integrate $\theta$ out, we develop the quadratic form:
\begin{eqnarray*}
&(y-g_x\theta)^\top (I_n + k{\rm Corr})^{-1}(y-g_x\theta)&\\
&=& \\
&y^\top (I_n + k{\rm Corr})^{-1} y + \theta^\top g_x^\top (I_n + k{\rm Corr})^{-1} g_x\theta -2y^\top (I_n + k{\rm Corr})^{-1} g_x\theta&
\end{eqnarray*}
Identifying the above expression with the development of $(\theta - \widehat\mu_1)^\top \widehat\Sigma_1^{-1}(\theta - \widehat\mu_1)$, we find that the conditional posterior distribution of $\theta$ given $\lambda$, is proper, and given by:
\begin{eqnarray*}
\theta |y, \lambda, k, \mathfrak M_1, x &\sim& \mathcal N(\widehat\mu_1, \widehat\Sigma_1),
\end{eqnarray*}
where:
$$
\widehat\mu_1 = (g_x^\top V_k^{-1} {g_x})^{-1} {g_x}^\top V_k^{-1} y;\quad \widehat\Sigma_1 = ({g_x}^\top V_k^{-1} g_x)^{-1},
$$
with $V_k = (I_n + k{\rm Corr}).$ By integrating $\theta$ out, we obtain the product of the partially integrated likelihood by the prior for $\lambda$ and $k$:
\begin{eqnarray}
&[ y | \lambda, k, \mathfrak M_1, x ] [ \lambda] &\nonumber\\
&=& \\
&\frac{\lambda^{-(n-p)-1}}{\sqrt{2\pi}^{n-p}} |V_k|^{-1/2} |\widehat \Sigma_{1,k}|^{1/2}\exp\left\{ -\frac{1}{2\lambda^2}y^\top(V_k^{-1} - V_k^{-1}g_x(g_x^\top V_k^{-1} g_x)^{-1} g_x^\top V_k^{-1})y\right\}.&\nonumber
\label{eq:M1_lambda_marginal}
\end{eqnarray}
Again, we can show that the above quadratic form  is equal to the vector norm:
\begin{eqnarray*}
y^\top(I_n - g_x(g_x^\top\!g_x)^{-1} g_x^\top)y &=& ||V_k^{-1/2}(y-g_x \widehat\mu_1)||^2.
\end{eqnarray*}
This is strictly positive with probability $1$.

Next step is to integrate (\ref{eq:M1_lambda_marginal}) over $\lambda$, using the change in variables $\tau = 1/{\lambda^2}$, yielding:
\begin{eqnarray*}
[ y | \tau,k, \mathfrak M_1, x ] [ \tau] 
&=& \frac{1}{2}\frac{\tau^{\frac{n-p}{2}-1}}{\sqrt{2\pi}^{n-p}}  \sqrt{\frac{|\widehat\Sigma_1|}{|V_k|}}\exp\left\{ -\frac{||V_k^{-1/2}(y - g_x \widehat\mu_1)||^2}{2}\tau\right\}.
\end{eqnarray*}
We recognize the $\mathcal Ga\left(\frac{n-p}{2}, \frac{||V_k^{-1/2}(y - g_x \widehat\mu_1)||^2}{2}\right)$ Gamma distribution, so that :
\begin{eqnarray*}
[ y | k, \mathfrak M_1, x ] 
&=& \frac{1}{2}\frac{\Gamma(\frac{n-p}{2})}{\sqrt{\pi}^{n-p}} \sqrt{\frac{|\widehat\Sigma_1|}{|V_k|}}
	||V_k^{-1/2}(y - g_x \widehat\mu_1)||^{-(n-p)} .
\end{eqnarray*}

Since this last expression is bounded in $k$, and the prior on $k$ is proper, we can conclude that $m_1(y) < \infty\quad \square$

\end{document}